\documentclass[runningheads]{llncs}
\usepackage[T1]{fontenc}
\usepackage{amsmath,amssymb,amsfonts}
\usepackage{graphicx}
\usepackage{pifont}
\usepackage{graphicx,subfigure}
\newtheorem{observation}{Observation}
\usepackage[ruled,vlined,linesnumbered]{algorithm2e}
\begin{document}
\title{Time Optimal Distance-$k$-Dispersion on Dynamic Ring}
%
%
\author{Brati Mondal\inst{1} \and
Pritam Goswami\inst{2}\and
Buddhadeb Sau\inst{1}}
\authorrunning{Brati Mondal, Pritam Goswami and Buddhadeb Sau}
%
\institute{Jadavpur University, Jadavpur, Kolkata, 700032 \and
Sister Nivedita University, Kolkata-700156
\email{bratim.math.rs@jadavpuruniversity.in, pgoswami.cs@gmail.com, buddhadeb.sau@jadavpuruniversity.in}}
\maketitle              
\begin{abstract}
Dispersion by mobile agents is a well studied problem in the literature on computing by mobile robots. In this problem, $l$ robots placed arbitrarily on nodes of a network having $n$ nodes are asked to relocate themselves autonomously so that each node contains at most $\lfloor \frac{l}{n}\rfloor$ robots. When $l\le n$, then each node of the network contains at most one robot. Recently, in NETYS'23, Kaur et al. introduced a variant of dispersion called \emph{Distance-2-Dispersion}. In this problem, $l$ robots have to solve dispersion with an extra condition that no two adjacent nodes contain robots.

In this work, we generalize the problem of Dispersion and Distance-2-Dispersion by introducing another variant called \emph{Distance-$k$-Dispersion (D-$k$-D)}. In this problem, the robots have to disperse on a network in such a way that shortest distance between any two pair of robots is at least $k$ and there exist at least one pair of robots for which the shortest distance is exactly $k$. Note that, when $k=1$ we have normal dispersion and when $k=2$ we have D-$2$-D. Here, we studied this variant for a dynamic ring (1-interval connected ring) for rooted initial configuration. We have proved the necessity of fully synchronous scheduler to solve this problem and provided an algorithm that solves D-$k$-D in $\Theta(n)$ rounds under a fully synchronous scheduler. So, the presented algorithm is time optimal too. To the best of our knowledge, this is the first work that considers this specific variant.

\end{abstract}
\section{Introduction}
\label{Section:1}
In the study of distributed systems, researchers have shown significant interest in employing a swarm of mobile computing entities, known as robots, to perform fundamental tasks. These robots, which are simple and inexpensive, collaborate to execute tasks on either discrete domains (such as graphs) or continuous domains (such as planes). Their tasks may include pattern formation, exploration, gathering, and dispersion.  The main motive of this research in this direction is to use limited-capability robots to achieve a robust goal in a distributed way with large-scale practical applications.

In this work, we are interested in the problem of \textit{Dispersion}. In this problem, the aim is to reposition $l$  robots on a graph with $n$ nodes in such a way that each node contains at most $\lfloor \frac{l}{n} \rfloor$ robots. So when $l\le n$, each node must contain at most 1 robot. 
The dispersion problem was first introduced by Augustine and Moses Jr. (\cite{AM2018ICDCN}). This problem is closely related to several extensively studied robot coordination challenges, such as scattering \cite{BFES2011SCATTERINGGRID}, covering, exploration \cite{DDKPU201537EXPLORATION}, self-deployment \cite{EB20117DEPLOYMENT}, and load balancing. The concept of dispersion has numerous practical applications, including relocating self-driving vehicles to recharge stations (where the vehicles act as mobile computing entities with the capability to locate these stations), deploying robots in areas that are inaccessible to humans, and other tasks involving exploration \cite{LDFS2016ICDCS}, gathering \cite{LFPPSV2020TCS}, and uniform distribution.

The dispersion problem has been studied under different communication models such as the local communication model and global communication model. In the local communication model robots can communicate with each other only when they reside on the same node whereas in the global communication model robots can communicate with other robots residing on different nodes although they are not aware of the graph structure. The global communication model seems better than the local one though, but it is not so easy to find the exact location of a particular node just by communicating with the node. So just like the local communication model, robots have to explore through edges in the global communication model. There is a huge existing literature on dispersion problem considering local communication as well as global communication \cite{AAMKS2018ICDCN,AM2018ICDCN,CKMS2023CALDAM,DBS2021CALDAM,GKM2024CALDAM,GMMP2024JPDC,KA2019ICDCN,KM2023NETYS,KMS2019ALGO,KMS2020ICDCN,KMS2020ICDCS,KMS2020WALCOM,KMS2022JPDC,KS2021OPODIS,MM2019TAMC,MMM2020ALGOSENSORS,MMM2021IPDPS,MMM2021TCS,PSM2021ICDCN,SSKM2020SSS}.
The problem considered here is a special type of dispersion problem, called \textit{Distance-$k$-Dispersion} aka D-$k$-D where the robots disperse themselves autonomously in such a way so that the distance between any two occupied nodes is at least $k$ and there is at least one pair of occupied nodes whose distance is exactly $k$, where occupied nodes contain exactly one robot. For $k=1$, the problem is equivalent to the dispersion problem. Also, Distance-$2$-Dispersion (introduced in \cite{KM2023NETYS}) problem is a special case of D-$k$-D where $k=2$. Another intriguing aspect of D-$k$-D is its application in maximizing the sensing area with the use of the fewest agents possible assuming the sensing distance of an agent is up to $k$ hop.

In this paper, we have studied the D-$k$-D problem on a specific graph topology which is a 1-interval connected ring having $n$ nodes, assuming the initial configuration is a rooted configuration with $l \le \lfloor \frac{n}{k}  \rfloor$ mobile robots.

\section{Related Works and Our Contribution}
\label{Section:2}
\subsection{Related work}
Dispersion of mobile robots is a hugely studied problem in swarm robotics. This problem was first introduced by Augustine and Moses Jr. \cite{AM2018ICDCN} considering different types of graphs like path, ring, tree, rooted graph, and arbitrary graph, where they analyze the trade-off between time and memory with local communication model. They provided the lower bound of memory required for robots ($\Omega(\log n )$) and the minimum run time ($\Omega(D)$), where $n$ is the number of nodes of the graph and $D$ is the diameter of the graph. For the rooted graph and arbitrary graph, the provided algorithms solve within $O(m)$ rounds with memory per robot $O(\log n)$ and $O(n \log n)$ respectively, where $m$ is the number of edges. 

There are many works of dispersion on arbitrary graph considering local communication model in synchronous scheduler \cite{AM2018ICDCN,KA2019ICDCN,KMS2019ALGO,SSKM2020SSS,KS2021OPODIS}. In \cite{KA2019ICDCN} Kshemkalyani and Ali proposed five algorithms to solve dispersion on arbitrary anonymous graph considering the local communication model for synchronous and asynchronous model. The first three algorithms require runtime $O(m)$ running time with $O(l \log \Delta)$ bits memory at each robot, where $m$, $l$ and $\Delta$ are the number of edges, number of robots and degree of the graph respectively. The last two algorithms consider the asynchronous scheduler, one solves dispersion $O(D \log \Delta)$ bits memory per robot with run time $O(\Delta ^ D)$ and another solves with time complexity $O((m-n)l)$ considering $O(\max(\log l,\log \Delta))$ memory per robot. Later in \cite{KMS2019ALGO} Kshemkalyani et al. improved the algorithm by improving time complexity to $O(\min (m,l \Delta) .\log l)$ with $O(\log n)$ bits memory for each robot. Then in \cite{SSKM2020SSS}, Shintaku et al. considered dispersion with same memory $O(\log (l+ \Delta))$ and time complexity as \cite{KMS2019ALGO} without prior knowledge of $m,l,\Delta$. In \cite{KS2021OPODIS}, Kshemkalyani et al. provided an improved algorithm considering the same assumptions with $O(\min (m,l \Delta))$ runtime.  

The previous works give deterministic solutions. Using randomized solutions some works reduced the memory of the robots. In the work \cite{RSSS2019SPAA}, the authors studied the time to achieve dispersion using random walks on different graphs. In \cite{MM2019TAMC}, Molla et al. showed the technique to break $\Omega(\log n)$ memory lower bound used for deterministic solution. They provided two solutions with $O(\Delta)$ and $O(\max(\log \Delta, \log D))$ memory per robot for the rooted graph. Later in \cite{DBS2021CALDAM} Das et al. provided an improved solution from \cite{MM2019TAMC} considering $O(\log \Delta)$ bits of memory per robot which is optimal. 

Recently dispersion problem has been studied for the configuration with faulty robots. In this type, first Molla et al. studied the dispersion problem in \cite{MMM2020ALGOSENSORS,MMM2021TCS} on the anonymous ring with weak byzantine fault. Later in \cite{MMM2021IPDPS} authors provided different dispersion algorithms on graphs considering strong byzantine robots. Later in \cite{PSM2021ICDCN} Pattanayak et al. considered dispersion with some crash-prone robots from a rooted configuration. Recently in \cite{CKMS2023CALDAM} Chand et al. improved the time complexity without changing optimal memory requirement considering both rooted and arbitrary configuration.

In the study of dispersion problem global communication model has been considered in many tasks \cite{KMS2020ICDCN,KMS2020WALCOM,KMS2022JPDC}. In \cite{KMS2020ICDCN,KMS2022JPDC} Khemkalyani et al. studied dispersion with memory-time trade-off with global communication. Using global communication the better trade-off is possible than local communication. Again they considered the dispersion problem on the grid network in \cite{KMS2020WALCOM} using global communication model. They looked specially at square grids and came up with an algorithm that gives good trade-offs. Global communication can be exploited along with some other powers to solve dispersion in a dynamic graph.  In \cite{AAMKS2018ICDCN} Agarwalla et al. studied dispersion on dynamic ring subject to 1-interval connectivity and vertex permutation dynamism. In \cite{KMS2020ICDCS} Khsemkalyani et al. studied dispersion on dynamic graphs utilizing the knowledge of a node's neighborhood along with global communication considering 1-interval dynamism, even in the presence of crash faults.

Recently in \cite{KM2023NETYS} Kaur et al. considered a different kind of dispersion problem called Distance-2-Dispersion or D-2-D, where robots arbitrarily placed on a graph need to achieve a configuration such that settled robots occupy no two adjacent nodes. Moreover, an unsettled robot can settle on a node where already a settled robot exists if there is no vacant node fulfilling the extra constraints. Here dispersion is obtained using $l \ge 1 $ robots in $O(m \Delta)$ synchronous rounds with $O(\log n)$ memory per robot. Later in \cite{GKM2024CALDAM}, using a strong team of $l>n$ robots they improved the runtime to $O(m)$ from a rooted configuration, where $n,m,\Delta$ be the number of nodes, number of edges and maximum degree respectively. For arbitrary configuration, they provided an algorithm that solves D-2-D in $O(pm)$ rounds with $O(\log n)$ memory per robot, where $p$ is the number of nodes containing robots in the initial configuration.

\subsection{Our Contribution}

In this work, we aim to solve a different variant of the dispersion problem called Distance-$k$-Dispersion (D-$k$-D) on a dynamic ring. In the general dispersion problem, the goal is to reposition $l$ robots on the nodes of an $n$-node graph so that each node contains at most $\lfloor \frac{l}{n}\rfloor$ robots. For $l \le n$, each node of the graph contains at most one robot. But in the Distance-$k$-Dispersion problem, robots reposition themselves autonomously to achieve a configuration where the shortest distance between any two pair of robots is at least $k$ and there exists at least one pair of robots whose distance is exactly $k$. Note that, for $k=1$, the problem is similar to a normal dispersion problem, and for $k=2$, this problem is the same as Distance-2-Dispersion (introduced by \cite{KM2023NETYS}). So to the best of our knowledge, we consider Distance-$k$-Dispersion for the first time here on dynamic ring. Here we have assumed 1-interval connectivity dynamism where the adversary can omit at most one edge in a certain round from the ring without hampering the connectivity. 

Each robot has a unique identifier and initially, they all are on the same node i.e. the initial configuration is a rooted configuration. Robots have weak global multiplicity detection and strong local multiplicity detection. Robots work under a fully synchronous (\textsc{FSync}) scheduler. We have shown that a solution of D-$k$-D is not possible on a 1-interval connected dynamic ring for a semi-synchronous (\textsc{SSync}) scheduler if the initial configuration has a multiplicity node. Also, a time lower bound ($\Omega(n)$ rounds) is given for the D-$k$-D algorithm. Then we have provided a deterministic and distributed algorithm \textsc{D-$k$-D DynamicRing}, which solves the Distance-$k$-Dispersion problem for rooted initial configuration by the robots with chirality (i.e. the particular notion of orientation: clockwise or counter-clockwise). The algorithm \textsc{D-$k$-D DynamicRing} will terminate in $O(n)$ rounds. Observe that, if the robots have no chirality, then there are some known techniques in \cite{AAMKS2018ICDCN}, applying that chirality can be achieved so that the robots can agree on a particular direction (clockwise or counterclockwise). Then one can easily apply the existing algorithm \textsc{D-$k$-D DynamicRing} to achieve Distance-$k$-Dispersion even considering the set of robots without chirality.
We can also say that the algorithm \textsc{D-$k$-D DynamicRing} is time optimal as its runtime matches with the lower bound provided.

\section{Model and Problem 
Definition}\label{Section:3}
let $\mathcal{R}$ be a ring with $n$ consecutive nodes $v_0$,$v_1$,$v_2$,$\dots$,$v_{n-1}$ in a certain direction (either clockwise or counter clockwise direction), where the nodes $v_{i}$ is connected to both $v_{i-1 \pmod{n}}$ and $v_{i+1 \pmod{n}}$ by edges. These nodes are unlabeled.

\subsection{Robot Model:} 
let $\{r_1,r_2,\dots,r_l\}$ be a set of $l$ robots reside on the nodes of ring $\mathcal{R}$. The robots are identical i.e. physically indistinguishable and homogeneous i.e. the robots execute the same algorithm. Robots can reside only on the node but not on the path connecting two nodes. Each robot has a unique label (ID) which is a string of constant length. Robots have $\log n$ bits of memory for storing their Id, where $n$ is the number of nodes of the graph. But this memory is not updatable. So robot can't store anything else persistently. One or more robots can occupy the same node. If two or more robots are on the same node they are called co-located. The robots can communicate with each other only when they are co-located. In each round, co-located robots in a certain node can detect the minimum ID robot of that node.
\subsubsection{Visibility and communication}
Robots have the ability of strong local multiplicity detection which means the robot can detect how many robots are present in its node. Robots have weak global multiplicity detection i.e. robots can detect if there exists no robot, one robot, or more than one robot in a particular node within its visibility range. The visibility of a robot is defined as the number of hops the robot can see another node which is denoted as $\phi$. Here $\phi\ge n/2$ i.e. robots can see all the nodes of the dynamic ring. Robots can move one node to another node if the edge between them is not missing.
\subsubsection{Chirality}
The particular notion of orientation (clockwise or counterclockwise) of the robot in the ring is called chirality. Here clockwise (CW) and counterclockwise (CCW) notion of direction is used in the usual sense. 
\subsection{Scheduler}
The activation of the robot controlled by an entity is called scheduler. Here robots work under fully synchronous (\textsc{Fsync}) scheduler. In each round adversary changes in dynamism of the ring initially. In \textsc{Fsync} scheduler all robots get activated at the same time and perform look-compute-move synchronously.\\
\textbf{Look:} Robots take a snap of their views to collect information and communicate with the robots of their own node.\\ 
\textbf{Compute:} Robots perform computations on the basis of their look phase and determine whether to move or not and in which direction.\\
\textbf{Move:} The robots that are decided to move in compute phase, move to their determined direction. Robots can move in their determined direction only if the edge in that direction is not missing.

\subsection{Dynamic Ring Structure}
In this work, we have assumed the 1-interval connectivity dynamism. In $1$-interval connectivity (from \cite{LDFS2016ICDCS}) adversary can omit at most one edge from the ring in a certain round without losing the connectivity of the ring. 

A dynamic ring with $n$ nodes is considered here. Each node is anonymous i.e. they have no specific identifiers. Nodes have no memory. Initially, all the robots are on the same node, which is called \textbf{root node}. The ring is always connected. There are two types of nodes: \textbf{occupied node} and \textbf{null node}. If a node contains a robot then it is called \textbf{occupied node}. If there is no robot on the node then the node is called \textbf{null node}. According to the numbers of robots present on the node \textbf{occupied node} is of two types: \textbf{multiplicity node}, \textbf{singleton node}. \textbf{Multiplicity node} contains multiple number of robots and \textbf{singleton node} contains only one robot. The robot on a singleton node can be called \textbf{singleton robot}.

\subsection{Problem Definition}
Let $l$ robots reside on an $n$ node dynamic ring. Initially, the robots are co-located at a particular node. Without loss of generality let us assume that node is $v_0$. In the dispersion problem, the robots reposition themselves autonomously in such a way that at each node there is at most one robot. In \textit{distance-$k$-dispersion} (D-$k$-D) problem, the robots reposition to achieve a configuration that satisfies the following: \\
(i) The occupied node contains exactly one robot.\\
(ii) The distance between any two occupied nodes is at least $k$.\\
(iii) There exists at least one pair of nodes whose distance is exactly $k$.

\section{Impossibility Result and Lower Bound}\label{Section:4}
\begin{theorem}
    Distance-$k$-Dispersion on the dynamic ring is not possible for a semi-synchronous scheduler if the initial configuration has a multiplicity node on a 1-interval connected ring.  
\end{theorem}
\begin{proof}
If possible let, $\mathcal{A}$ be an algorithm that solves the Distance-$k$-Dispersion problem for semi-synchronous (\textsc{Ssync}) scheduler when the initial configuration has a multiplicity node. In semi-synchronous scheduler, time is also divided into rounds of equal duration like fully synchronous (\textsc{Fsync}) scheduler but a subset of robots get activated at the beginning of the round and execute look-compute-move cycle. The adversary has the power to choose if a certain robot is activated or remains idle at a certain round. Also, the adversary can omit at most one edge in a certain round from the ring without hampering the connectivity of the ring structure. 

Let from the beginning of execution of $\mathcal{A}$, adversary actives at most one robot from a multiplicity node, say $v$, in each round. 
 Let at a round $t$, a robot $r$ is activated on $v$. If according to $\mathcal{A}$, $r$ wants to move along the edge, say $e$ (incident to $v$) at round $t$ , adversary removes $e$ at round $t$. So $r$ can not move out from the multiplicity node. This is true for any $t \ge 0$ and any robot $r$ on $v$. So, no robot moves out from the multiplicity. So, Distance-$k$-Dispersion is impossible to solve on a 1-interval connected ring.  
\end{proof}

\begin{theorem} Any dispersion algorithm on a ring with $n$ nodes from a rooted configuration takes $\Omega(n)$ rounds.
\end{theorem}
\begin{proof}
Let $\mathcal{A}$ be a dispersion algorithm that terminates for any rooted initial configuration with $l=n$ robots. Let $r$ be the robot that moves furthest to a node, say $v_t$ from the root, say $v_0$, in the target configuration. Then in the worst case distance between $v_0$ and $v_t$ is $\frac{n}{2}$. Since a robot can move only one unit at a round, at least $\frac{n}{2}$ rounds are needed. So, $\mathcal{A}$, takes $\Omega(n)$ rounds to terminate. 
\end{proof}

\begin{corollary}
    Any dispersion algorithm on a 1-interval connected ring with $n$ nodes, takes $\Omega(n)$ rounds to terminate.
\end{corollary}

\begin{corollary}\label{th:f'}
    Any D-$k$-D algorithm on a 1-interval connected ring with $n$ nodes, takes $\Omega(n)$ rounds to terminate.
\end{corollary}

\subsection{Road map to the Paper}
A brief road map to the rest of the paper is described here. In Section~\ref{Section:3}, the problem definition along with the model is discussed. In Section~\ref{Section:4}, an impossibility result showing in which case the problem can't be solvable  is given and the asymptotic lower bound for the run time of the dispersion algorithm is also discussed. In Section~\ref{Section:5}, Some preliminaries such as some definitions, and notations have been established. Section~\ref{Section:6} is dedicated to describing the provided algorithm along with the correctness results. Finally, this work is concluded in Section~\ref{Section:7}.

\section{Definitions and Preliminaries}\label{Section:5}
In this section let us describe some useful definitions.

\begin{definition}[Configuration]
A configuration at a round $t$ can be denoted as $C(t):(V,t,f)$ where $f$ is a function from $V$ to $\{0,1,2\}$ such that
\begin{equation*}
f(v) = 
\begin{cases}
        0, & \text{when v is null node at round t}\\
        1, & \text{when v is single node at round t}\\
        2, & \text{when v is multiplicity node at round t}
\end{cases} 
\end{equation*}
    
\end{definition}

\begin{definition}[Arc]
 An induced sub-graph of the ring $\mathcal{R}$ containing the vertices starting from a node $v_x$ in the direction $\mathcal{D}$ ending at the node $v_y$ is called an arc and it is denoted as $(v_x,v_y)_{\mathcal{D}}$.   

 The number of nodes in an arc is called arc length.
\end{definition}

\begin{definition}[Distance between two nodes]
For any two nodes $v_x$ and $v_y$ on the ring $\mathcal{R}$, the distance between the nodes $v_x$ and $v_y$ in a direction $\mathcal{D}$ is the number of edges from the node $v_x$ to $v_y$ in direction $\mathcal{D}$ and it is denoted as $d_{\mathcal{D}}(v_x,v_y)$.
    
\end{definition}

\begin{definition}[Consecutive occupied nodes]
Let $v_x$ and $v_y$ be two occupied nodes on the ring $\mathcal{R}$. Then $v_x$ and $v_y$ are said to be consecutive occupied nodes (Fig~\ref{fig:occupied node}) if there exists a direction $\mathcal{D}$ (either clockwise or counterclockwise direction) such that $(v_x,v_y)_{\mathcal{D}}$ does not contain any other occupied node other than $v_x$ and $v_y$.
    
\end{definition}

\begin{figure}[h]
     \centering
     \includegraphics[height=3.3cm]{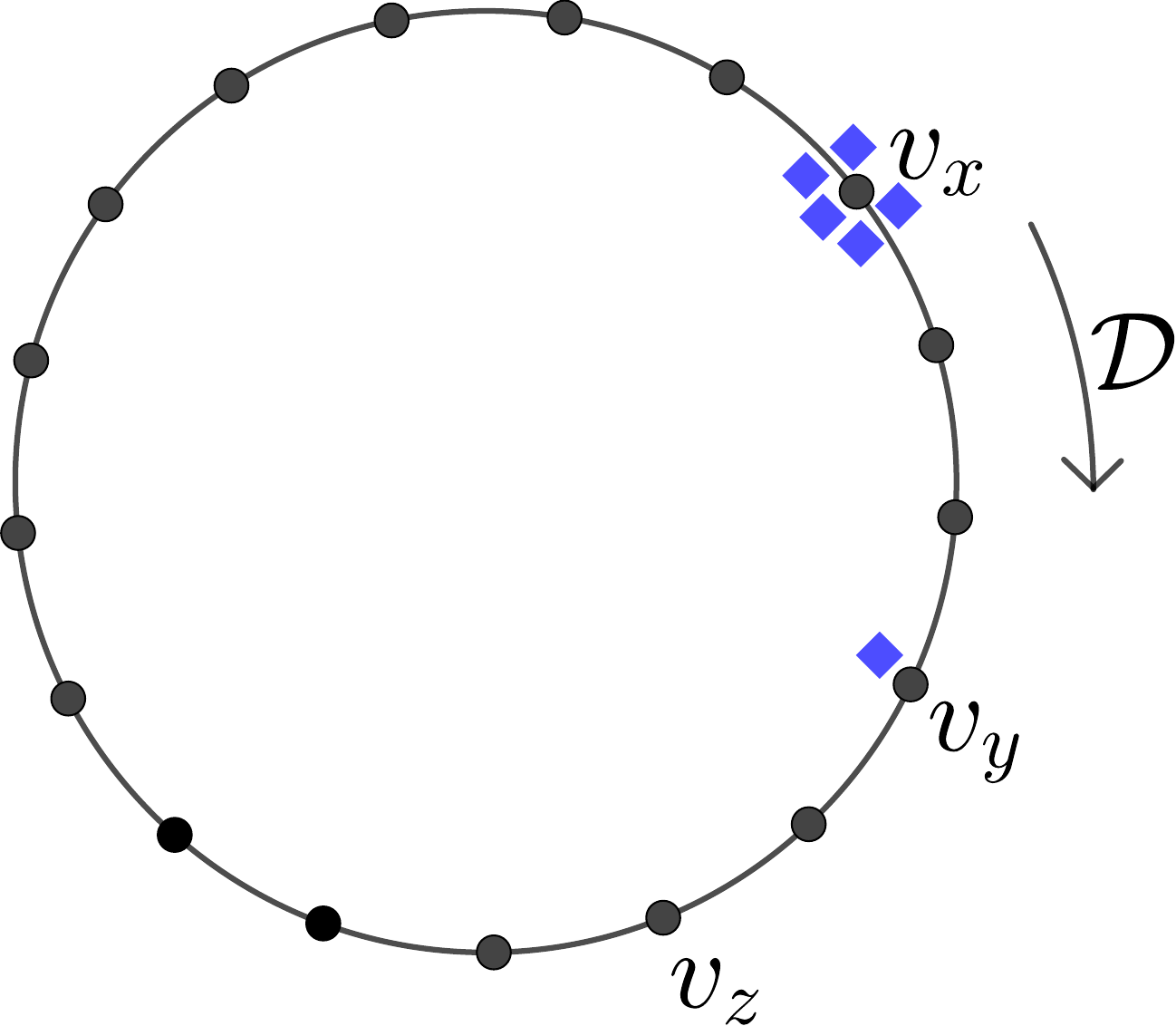}
     \caption{$v_x$ and $v_y$ are consecutive occupied nodes in the arc $(v_x,v_z)_{\mathcal{D}}$ of length $5$ in direction $\mathcal{D}$. Black dots represent the nodes on the ring and blue boxes represent the robots. }
     \label{fig:occupied node}
 \end{figure}

\begin{definition}[Clockwise Chain]
Let $\mathcal{R}$ be a ring with vertices $v_0,v_1,\dots,v_{n-1}$ in clockwise direction $\mathcal{D}$ (say). The arc starting from the node $v_{i-1 \pmod{n}}$ to $v_j$ i.e. $(v_{i-1 \pmod{n}},v_j)_{\mathcal{D}}$, for some $i,j$ is called clockwise chain (CW-chain) if

(i) $v_i$ is multiplicity node, $v_{j-1}$ is occupied node and $v_j$ is null node.

(ii) for any two consecutive occupied nodes $v_x$, $v_y$  $\in (v_{i-1 \pmod{n}},v_j)_{\mathcal{D}}$, $d_{\mathcal{D}}(v_x,v_y)=k$.

(iii) $d_{\mathcal{D}}(v_{j-1},v_p)>k$, for the nearest occupied node $v_p$ in clockwise direction from the node $v_{j-1}$.


similarly one can define counterclockwise chain (CCW-chain) considering counterclockwise direction.

\end{definition}

\begin{definition}[Chain Configuration]\label{defn:chain config}
Chain configuration is a configuration that contains both clockwise and counterclockwise chains (Fig~\ref{fig:chain config}).

\end{definition}

\begin{figure}[h]
     \centering
     \includegraphics[height=3.8cm]{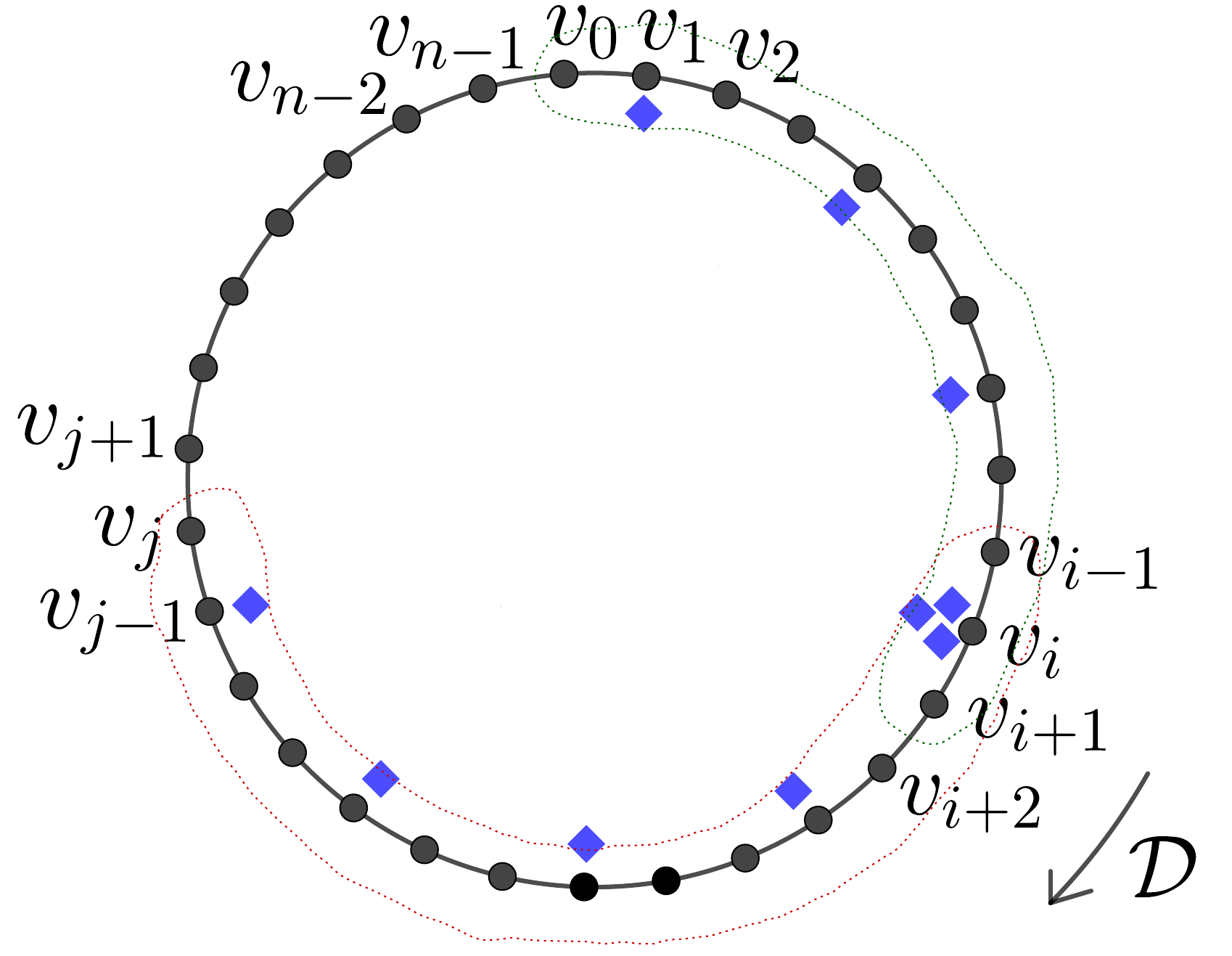}
     \caption{Chain configuration where $(v_{i-1},v_j)_{\mathcal{D}}$ is a chain in clockwise direction $\mathcal{D}$ and $(v_{i+1},v_0)_{\mathcal{D}'}$ is a chain in counter clockwise direction for $k=3$. }
     \label{fig:chain config}
 \end{figure}

\begin{figure}[!h]

\centering     
\subfigure[Both Chain good chain ]
{\label{fig:both good chain}
\includegraphics[height=3.2cm]{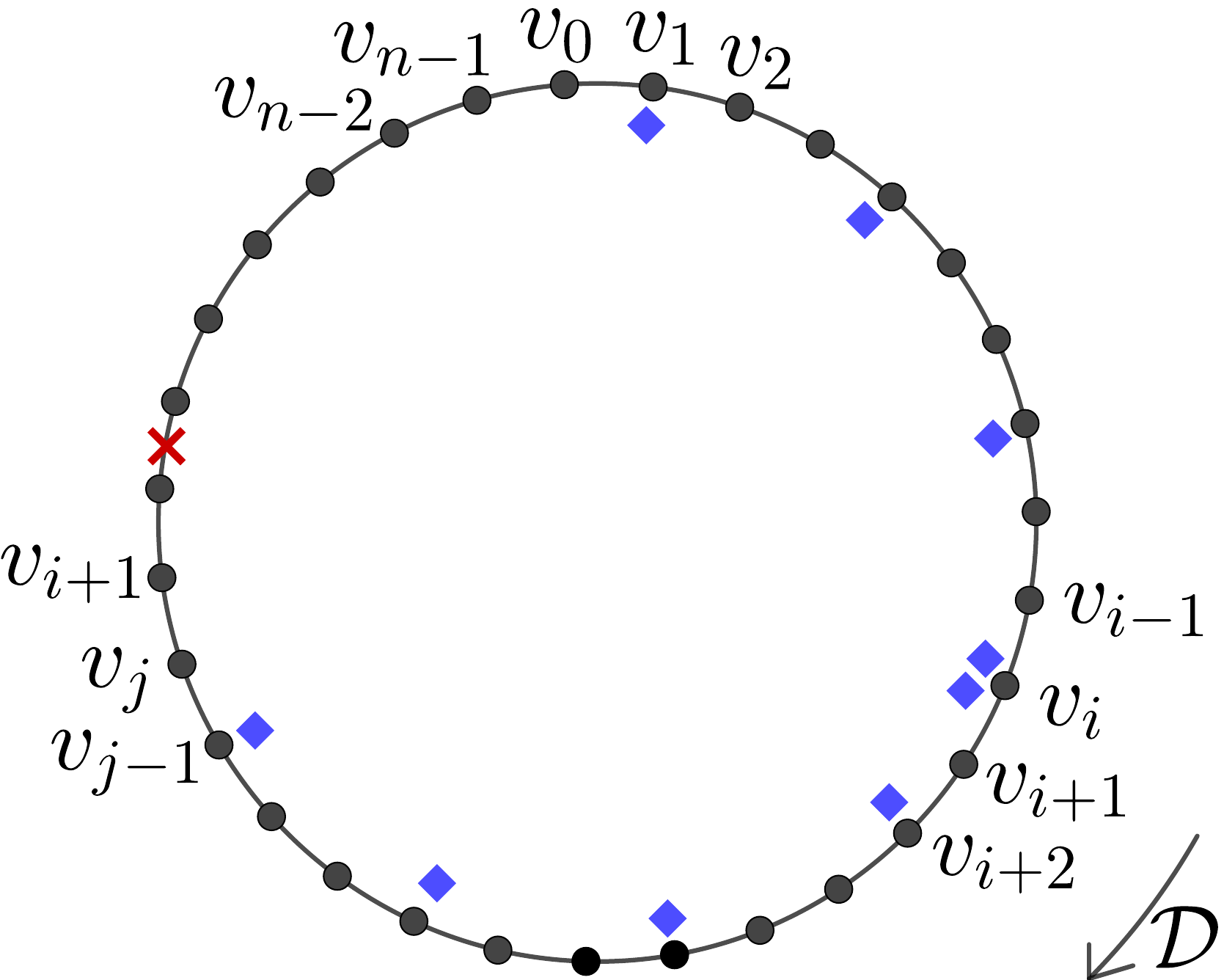}}
\hfill
\subfigure[one good and one bad chain]
{\label{fig:good chain bad chain}
\includegraphics[height=3.2cm]{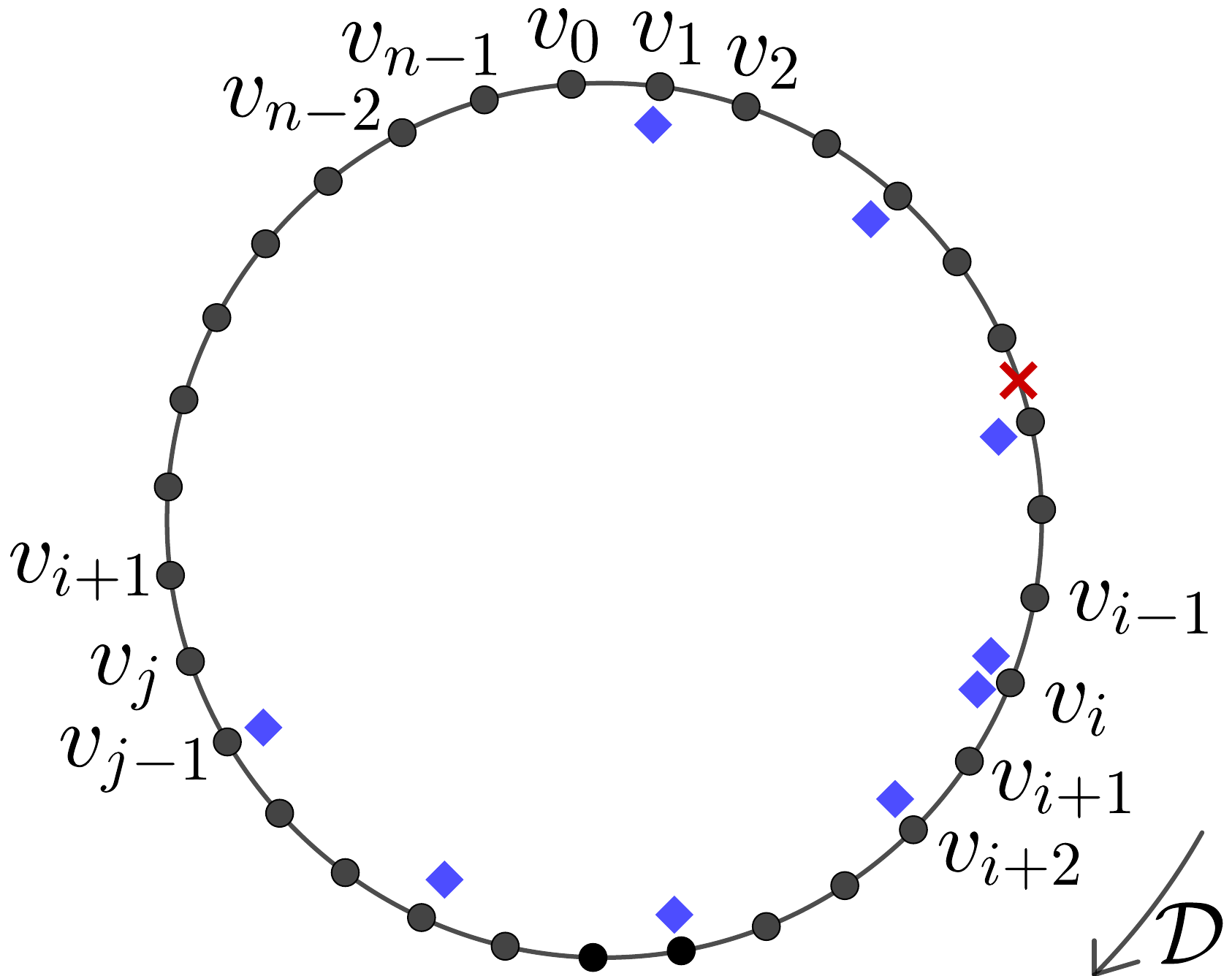}}
\hfill
\subfigure[Both chain bad chain]
{\label{fig:both bad chain}
\includegraphics[height=3.2cm]{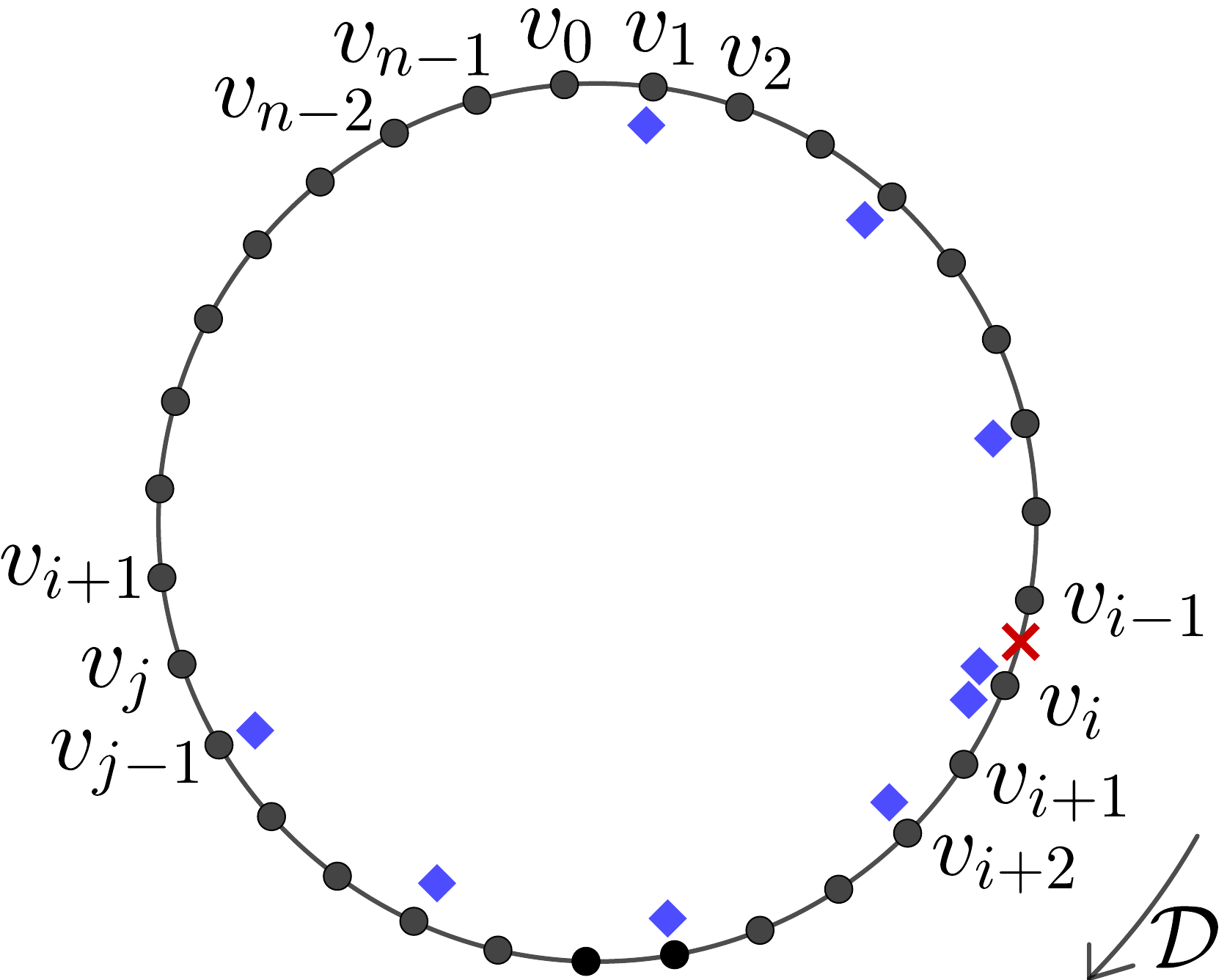}}
\caption{Different types of chain configuration for $k=3$}
\label{fig:diffchainconfig}
\end{figure}

\begin{definition}[Range of the Chain Configuration]
In a chain configuration, the maximum size arc of the ring $\mathcal{R}$ such that every node on that arc is contained in at least one chain is called the range of the chain configuration.  
\end{definition}

If the configuration is not chain configuration then it is called \textit{non-chain configuration}.

 In a chain configuration, a chain is called a \textit{good chain} if the chain contains no missing edge. Otherwise, it is called a \textit{bad chain}. There are three types of chain configuration: chain configuration containing two good chains, containing one good chain, and one bad chain, containing two bad chains.

\begin{definition}[Dispersed Configuration]
A configuration is called dispersed configuration if there is no multiplicity node.
\end{definition}

\begin{definition}[Head]
In a dispersed configuration, if there exists exactly one pair of occupied nodes $v_x$ and $v_y$ on the ring $\mathcal{R}$ such that
$d_{\mathcal{D}}(v_x,v_y)<k$ for clockwise direction $\mathcal{D}$, then $v_x$ is called the head and denoted as $H$.
\end{definition}

\begin{definition}[Block]\label{defn:block}
Let $\mathcal{R}$ be a ring with vertices $v_0,v_1,\dots,v_{n-1}$ in a direction $\mathcal{D}$. The arc starting from the node $v_{i}$ to $v_j$ i.e. $(v_i,v_j)_{\mathcal{D}}$, for some $i,j$ with maximum arc length is called a block in direction $\mathcal{D}$ if

(i) $v_i$ is a multiplicity node or head.

(ii) $v_j$ is a null node and $v_{j-1}$ is occupied node, where $v_{j-1}$ is the node adjacent to $v_j$ in the arc $(v_i,v_j)_{\mathcal{D}}$.

(iii)(a) $d_{\mathcal{D}}(v_x,v_y)=k$, for all consecutive occupied nodes $v_x$, $v_y\in (v_i,v_j)_{\mathcal{D}}$

or,

(b) $d_{\mathcal{D}}(v_i,v_z)<k$, where $v_i$ and $v_z$ are consecutive occupied nodes $\in (v_i,v_j)_{\mathcal{D}}$ and $d_{\mathcal{D}}(v_x,v_y)=k$, for all consecutive occupied nodes $v_x$, $v_y$ $\in (v_z,v_j)_{\mathcal{D}}$

(iv) $d_{\mathcal{D}}(v_{j-1},v_p)>k$, for the nearest occupied node $v_p$ in direction $\mathcal{D}$ from the node $v_{j-1}$.

The node $v_i$ is called block head. 
\end{definition}
For a block conditions (iii)(a) and (iii)(b)
 can not satisfy simultaneously. In a block if condition (iii)(a) is satisfied it is called a \textit{chain block} otherwise if condition (iii)(b) is satisfied it is called a \textit{non-chain block}.

\begin{figure}[h]

\centering     
\subfigure[Block configuration. ]
{\label{fig:block}
\includegraphics[height=3.2cm]{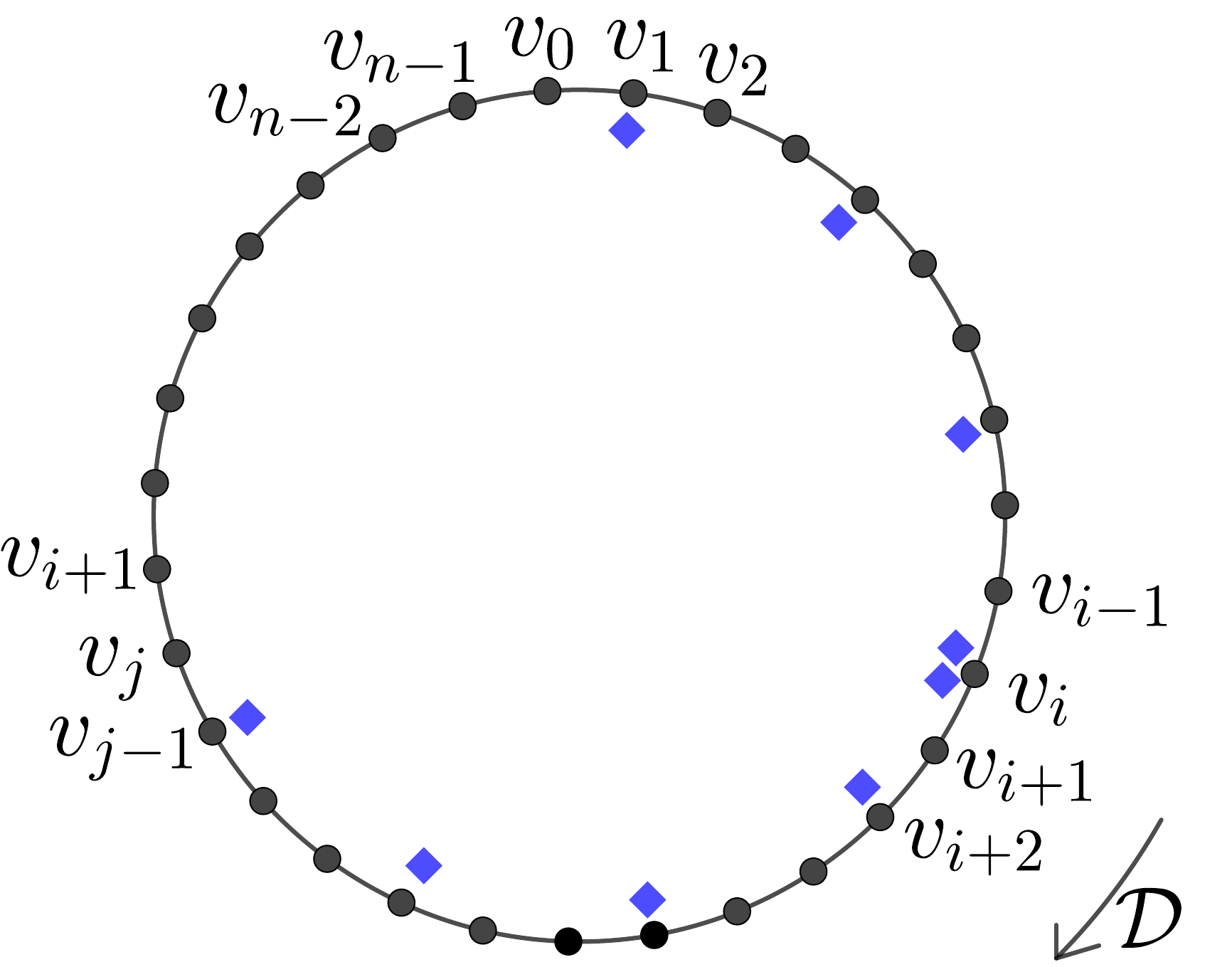}}
\hspace{0.1cm}
\subfigure[Dispersed configuration.]
{\label{fig:dispersed config}
\includegraphics[height=3.2cm]{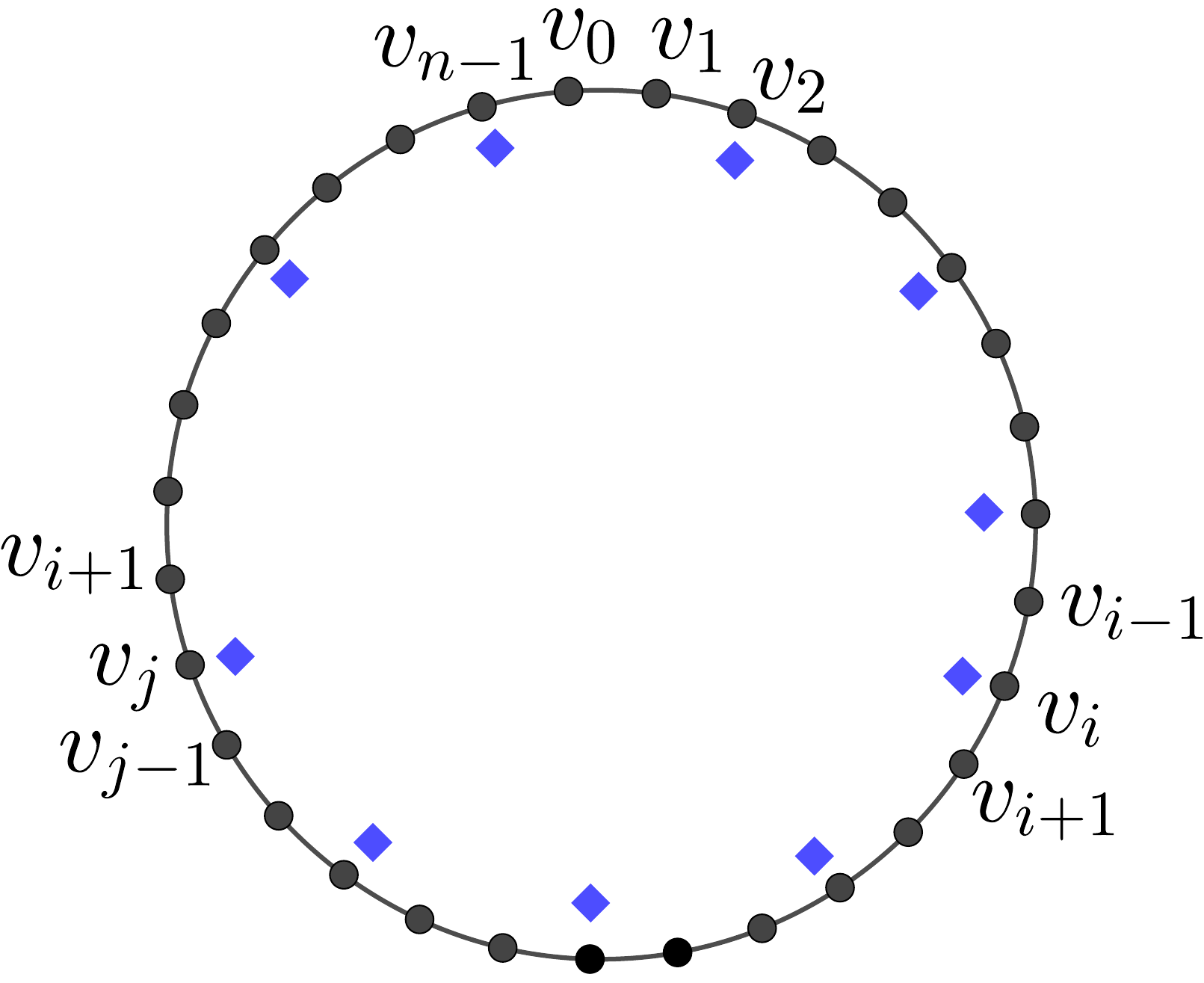}}
\caption{Different types of non-chain configuration for $k=3$}
\label{fig:diffnonchainconfig}
\end{figure}

\begin{definition}[Target Configuration]
A dispersed configuration is said to be target configuration ($C_T$) if,
$d_{\mathcal{D}}(v_x,v_y)\ge k$, for all consecutive occupied nodes $v_x$, $v_y$ on the ring $\mathcal{R}$, where $\mathcal{D}$ be either CW direction and CCW direction. 
\end{definition}

\begin{figure}[h]
     \centering
     \includegraphics[height=3.2cm]{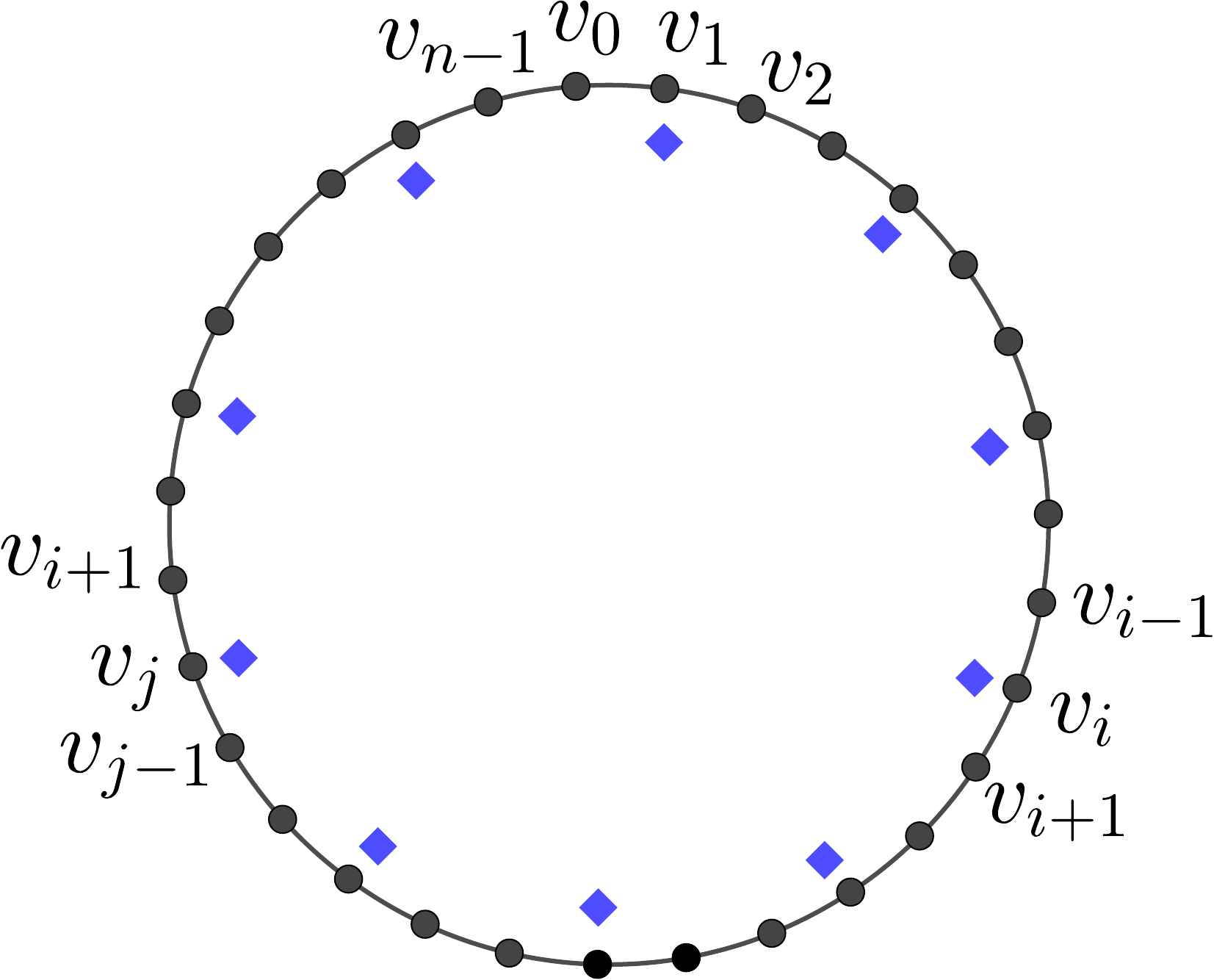}
     \caption{Target Configuration of D-$k$-D where $k=3$}
     \label{fig:target configuration}
 \end{figure}

\section{Algorithm for rooted configuration}\label{Section:6}
\subsection{Algorithm considering chirality}
\subsubsection{Description of the algorithm}
In this section, the main idea of D-$k$-D has been discussed. All the robots are initially placed on a particular node i.e. initially the configuration is rooted. The robots have chirality i.e. they agree on a particular direction (without loss of generality let the direction be clockwise direction).

The main task is divided into two sub-tasks. If the configuration is a chain configuration (Fig.~\ref{fig:chain config}), then a robot, say  $r$ on chain executes the algorithm \textsc{Spread($r$)} (Algorithm~\ref{Spread}). If the configuration is not a chain configuration (Fig.~\ref{fig:block}) then the robots reconstruct the chain configuration by executing \textsc{ReconstructChain($r$)} (Algorithm~\ref{ConstructChain}).

The main aim is to form the target configuration ($C_T$) (Fig.~\ref{fig:chain config}) where the distance between two consecutive occupied nodes is at least $k$. The brief idea of the algorithm \textsc{D-$k$-D DynamicRing($r$)}, proposed here is as follows. 
Note that initially, the rooted configuration is a chain configuration. The idea is to spread the range of the chain configuration by relocating one more agent from the multiplicity node to a null node (thus number of robots on multiplicity decreases). Now since the ring is 1-interval connected, to achieve a spread maintaining the chain configuration might not be possible. This ends up resulting in a non-chain configuration (Fig~\ref{fig:chain to nonchain}). In this scenario, the aim becomes reconstructing a chain configuration from the non-chain configuration while maintaining the number of robots on the multiplicity (Fig~\ref{fig:non-chain to chain}). So a combination of \textsc{Spread($r$)} followed by finite consecutive execution (at most $k-1$) of  \textsc{ReconstructChain($r$)}, the configuration spreads the range of the chains while reducing the number of robots on the multiplicity point by one. 
We call an execution of \textsc{Spread($r$)} followed by a finite consecutive (at most $k-1$) execution of \textsc{ReconstructChain($r$)} until a chain configuration is reached, an \textit{iteration}. Note that an iteration can be at most of $k$ rounds. So for $l$ robots within $l$ iteration, the configuration will be dispersed. Then from the disperse configuration within at most $k-1$ execution of \textsc{ReconstructChain($r$)} the configuration becomes $C_T$ (Fig.~\ref{fig:target configuration}).

\begin{figure}
\centering     
\subfigure[Chain configuration. ]{
\includegraphics[scale=0.21]{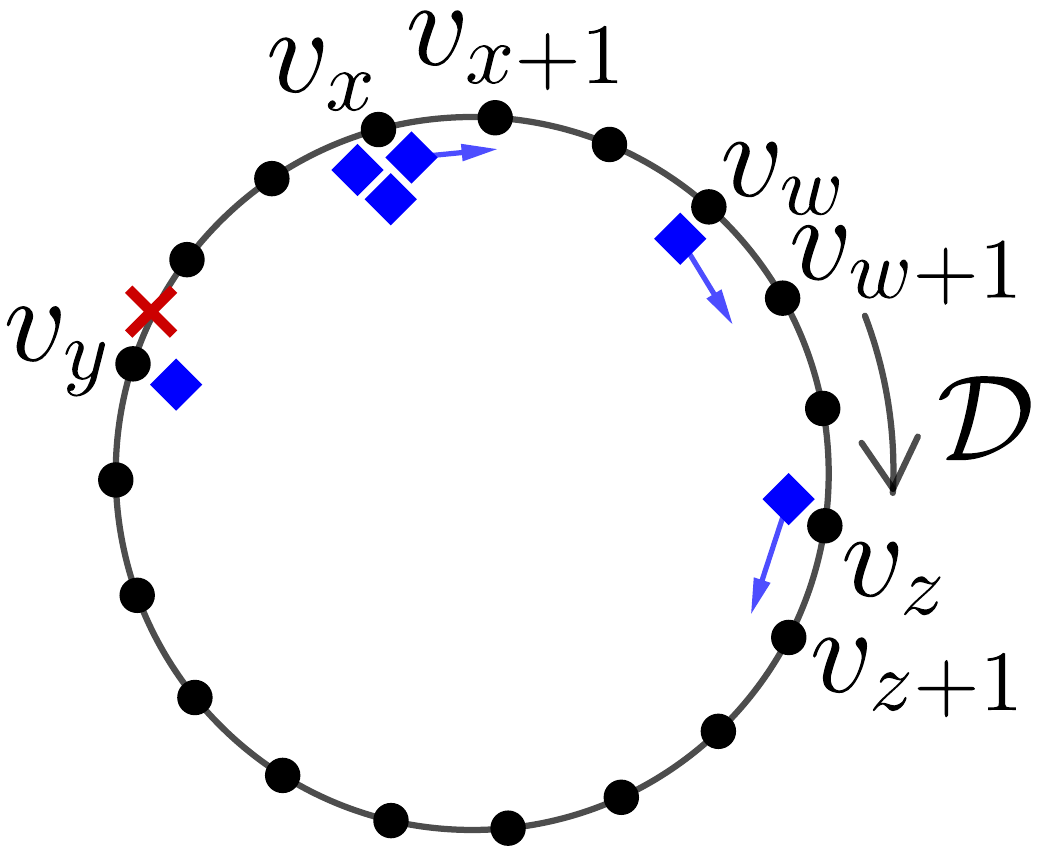}}
\hspace{0.1cm}
\subfigure[Non-chain configuration.]{
\includegraphics[scale=0.21]{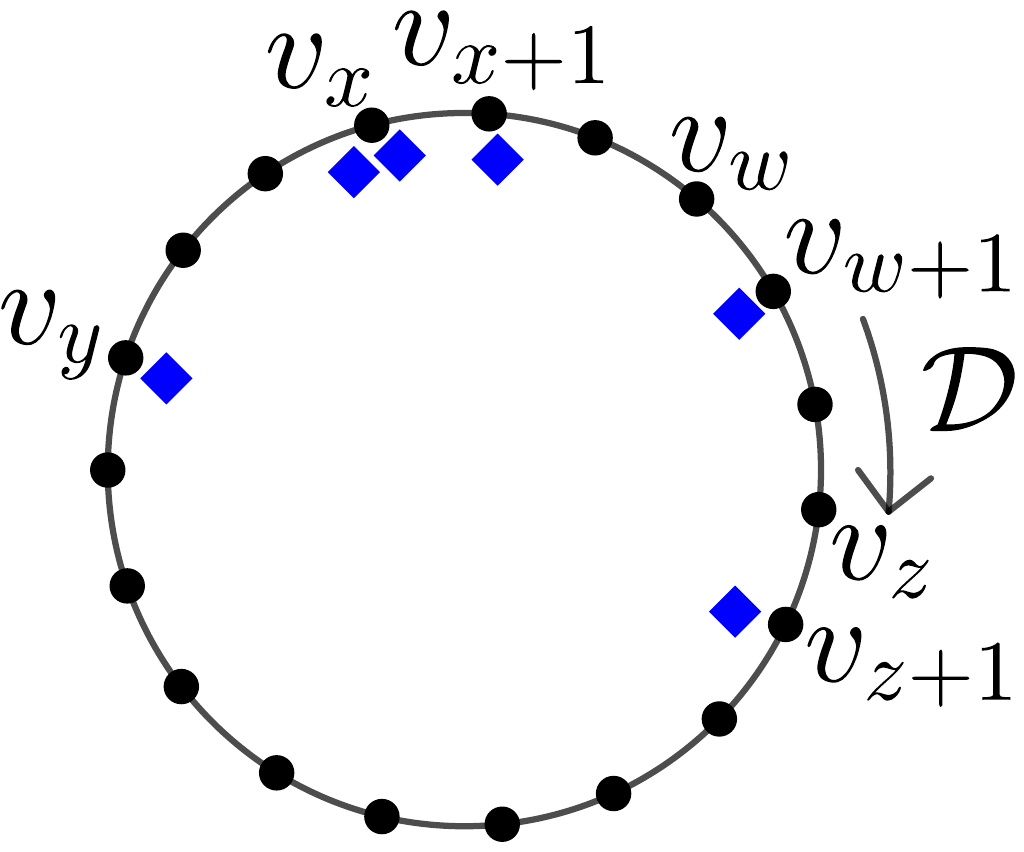}}
\caption{Chain to non-chain configuration}
\label{fig:chain to nonchain}
\end{figure}

\begin{figure}
\centering     
\subfigure[Non-chain configuration.]{
\includegraphics[scale=0.21]{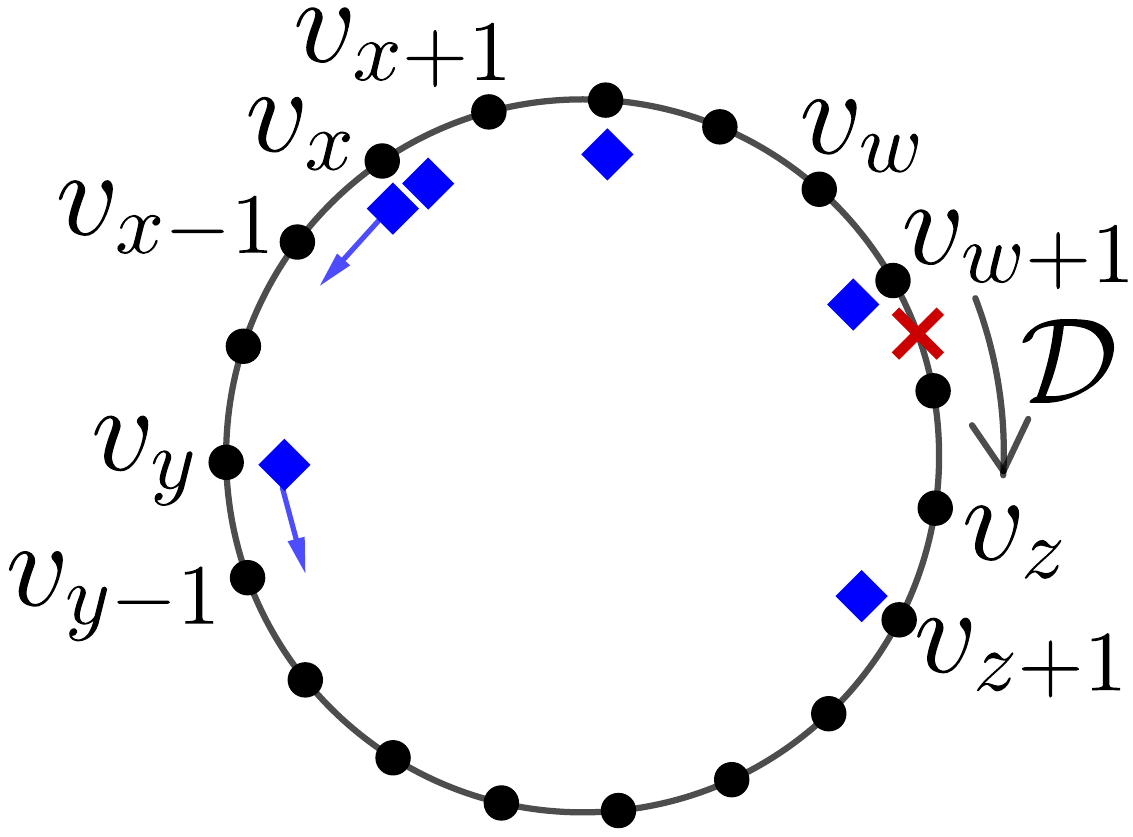}}
\hspace{0.1cm}
\subfigure[Chain configuration.]{
\includegraphics[scale=0.21]{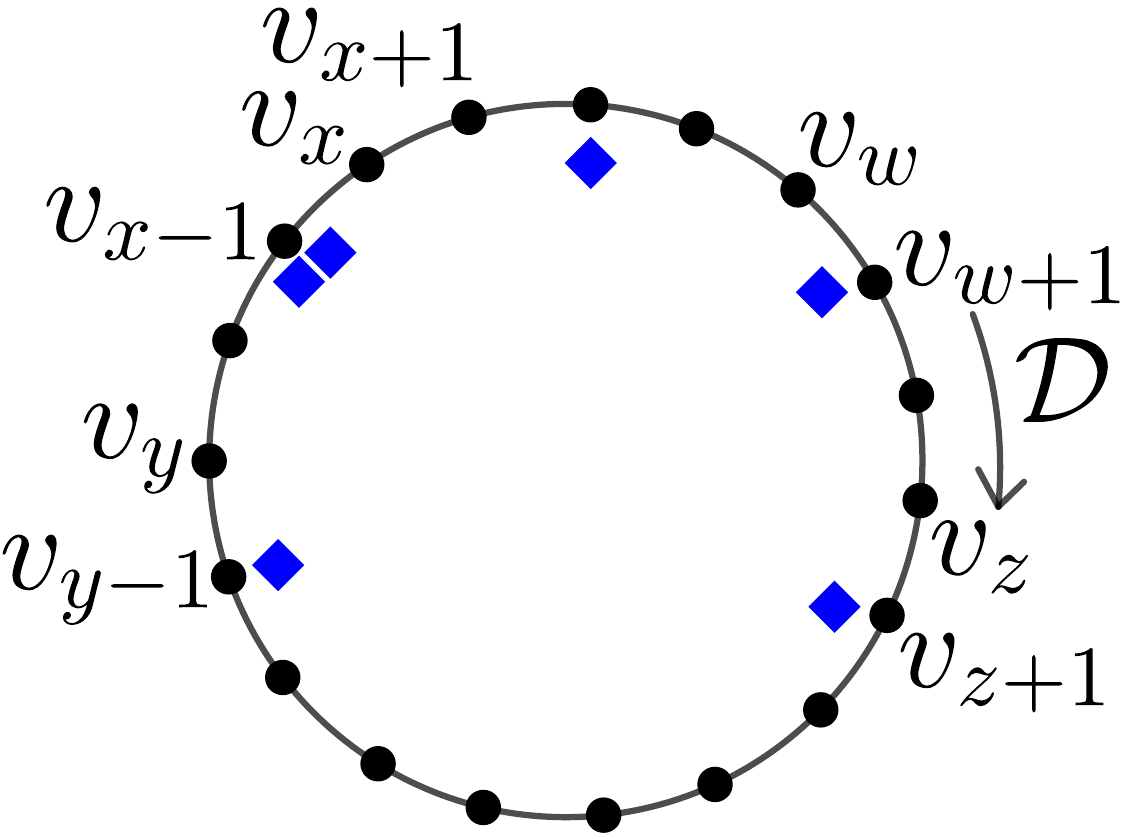}}
\caption{Non-chain to chain configuration}
\label{fig:non-chain to chain}
\end{figure}


\begin{algorithm}[ht]
\caption{\textsc{D-$k$-D DynamicRing($r$) }}
\label{DkD_DYNAMICRING}
\LinesNumbered

 \uIf{chain configuration}
    {Execute \textsc{Spread($r$)}\;}
 \Else{Execute \textsc{ReconstructChain($r$)}\;}   
 
\end{algorithm}

\begin{algorithm}[ht]
\caption{\textsc{Spread($r$)}}
\label{Spread}
\LinesNumbered
\eIf{$r$ is on multiplicity node}
{
    \If{$r$ has least ID among all nodes on the same node}
    {
            \uIf{both chain are good chain}
            {
                $\mathcal{D} \xleftarrow{}$ counter-clockwise direction\;
            }
            \uElseIf{exactly one chain is a bad chain}
            {
                 $\mathcal{D} \xleftarrow{}$  direction of the good chain\;
            }
            \ElseIf{both chains are bad chain}
            {
                 $\mathcal{D} \xleftarrow{}$  direction in which multiplicity node has no adjacent missing edge\;
            }

             $r$ moves in direction $\mathcal{D}$\;

    }
}
{
        \uIf{both chain are good chain}
            {
               \If{$r$ is on counter-clockwise chain}
               {
                 $r$ moves in counter-clockwise direction\;
               }
            }
            \uElseIf{exactly one chain is a bad chain}
            {
                 $\mathcal{D} \xleftarrow{}$  direction of the good chain\;
                 \If{$r$ is on good chain}
               {
                 $r$ moves in direction $\mathcal{D}$\;
               }
                 
            }
            \ElseIf{both chains are bad chain}
            {
                 $\mathcal{D} \xleftarrow{}$  direction in which multiplicity node has no adjacent missing edge\;
                 \If{$r$ is on the chain in direction $\mathcal{D}$}
                 {
                    $r$ moves in direction $\mathcal{D}$\;
                 }
            }
}
\end{algorithm}

\underline{\textbf{Subroutine} \textsc{Spread($r$)}\textbf{:}} During the execution of the algorithm at any round, the configuration can be chain configuration. In chain configuration two chains exist, one is CW-chain another is CCW-chain. Then three cases can arise: both are good chains, one chain is good another is bad, and both are bad chains (Fig~\ref{fig:diffchainconfig}).

Initially, all the robots are at same node which is called the root i.e. the initial configuration is rooted configuration. Note that this configuration is a chain configuration by Definition~\ref{defn:chain config}. 

If the configuration contains two good chains, then either there is no missing edge or the missing edge is not in any chain. In this case, the robots on singleton nodes of CCW chain move in counterclockwise direction. The robots on the multiplicity can find the least ID robot by exploiting local communication. The robot with the least ID on the multiplicity moves in counter clockwise direction.

If the configuration contains one good chain and one bad chain, then the missing edge is on the bad chain but not adjacent to the multiplicity node. Let $\mathcal{D}$ be the direction (either CW or CCW) from multiplicity in which the good chain exists. In this case, all singleton robots on the good chain and the least ID robot on multiplicity move in direction $\mathcal{D}$.

If the configuration contains two bad chains, then the missing edge is adjacent to multiplicity which is contained in both chains. This missing edge must be next to the multiplicity node either in clockwise direction or in counterclockwise direction. Let $\mathcal{D}$ be the direction in which multiplicity has no missing edge. In this case, the least ID robot on multiplicity moves in direction $\mathcal{D}$. The singleton robots of the chain in $\mathcal{D}$ direction move in the direction $\mathcal{D}$.
As a consequence of the subroutine \textsc{Spread($r$)}, we have the following observation:

\begin{observation}
    If the chain configuration contains one good chain and one bad chain or two bad chains then, after executing the subroutine \textsc{Spread($r$)} the chain configuration may get hampered and the configuration will become non-chain configuration.
\end{observation} 

\underline{\textbf{Subroutine} \textsc{ReconstructChain($r$)} \textbf{:}} If the configuration is not chain configuration (Fig~\ref{fig:diffnonchainconfig}) then the configuration must contain one chain block and one non-chain block (See Lemma~\ref{lemma: non chain config contains two blocks: chain block and non-chain block}). Let $\mathcal{D}$ be the direction in which the chain block occurs.

If there is no missing edge in both blocks (chain block or non-chain block), then all singleton robots on the chain block move away from the block head in the direction $\mathcal{D}$ and all robots on block head moves towards the direction $\mathcal{D}$. If the missing edge is in the non-chain block, then the procedure is the same as above.

If the missing edge is in the chain block, then the singleton robots in the non-chain block move away from the block head in direction $\mathcal{D}'$, opposite direction of $\mathcal{D}$

\begin{algorithm}
\caption{\textsc{ReconstructChain($r$)}}
\label{ConstructChain}
\LinesNumbered
   \eIf{$r$ is on block head }
   {
        \If{no edges are missing $\lor$ missing edge is in non-chain block}
        {
            $r$ moves in the direction of the chain block
        }
   }
   {
         \uIf{no edges are missing $\lor$ missing edge is in non-chain block}
         {
            \If{$r$ is on chain block}
            {
                $r$ moves away from block head\;
            }
         }
         \ElseIf{missing edge is in chain block}
         {
            \If{$r$ is on non-chain block}
            {
                $r$ moves away from the blcok head\;
            }
         }
   }
\end{algorithm}

\subsubsection{Correctness of the algorithm}
In this section, we prove the correctness of our algorithm \textsc{D-$k$-D DynamicRing}. 
\begin{lemma}\label{lemma:singleton node increases in spread}
At any round $t$, if a robot $r$ executes the subroutine \textsc{Spread($r$)} then at the end of the round $t$ the number of singleton nodes will increase.
\end{lemma}
\begin{proof}
Let, the configuration $C(t)$ be a chain configuration at time $t$, where a robot $r$ of the chain executes the sub-routine \textsc{Spread($r$)}. Then there exist two chains. For a chain configuration, the distance between any two consecutive occupied nodes in the same chain is $k$. Let $r_0$ be the robot on the multiplicity point with the least ID. Then, $r_0$ moves either clockwise or counterclockwise (depending on the existence of an edge in that same round). We first show that $r_0$ does not create a new multiplicity. Let at round $t$, $r_0$ move to $v_i$ from $v_{i-1}$ in a direction, say $\mathcal{D}$. Then the chain in direction $\mathcal{D}$ must not contain any missing edge.  Now, if $r_0$ creates a new multiplicity then at round $t$, $v_{i}$ must contain a robot that does not move during round $t$. This is not possible as all the robots on the chain in the direction $\mathcal{D}$ move away from multiplicity. Also, note that no singleton robot moves to the multiplicity. This concludes that the number of singleton nodes increases after execution of \textsc{Spread($r$)} for any robot $r$.
\end{proof}

\begin{lemma}\label{lemma: non chain config contains two blocks: chain block and non-chain block}
While executing Algorithm~\ref{DkD_DYNAMICRING}, if $C(t)$ be a non-chain configuration but not target configuration at round $t$, then the configuration $C(t)$ must have two blocks (chain block and non-chain block).
\end{lemma}
\begin{proof}
Let $C(t)$ be a configuration at round $t$, which is a non-chain configuration other than the target configuration while executing the Algorithm ~\ref{DkD_DYNAMICRING}. Then $C(t-1)$ was either chain configuration or non-chain configuration. 

\textit{Case I:} $C(t-1)$ be a chain configuration. Then by definition~\ref{defn:chain config} distance between any two consecutive occupied nodes in the configuration $C(t-1)$ is equal to $k$. As the configuration $C(t-1)$ is chain configuration, then the robots execute the algorithm~\ref{Spread}. Note that, three cases may arise:

\textit{Case I(a):} In the chain configuration $C(t-1)$, both chains are good. In this case, by executing Algorithm~\ref{Spread}, the least ID robot (say $r_0$) on the multiplicity node (say $v_0$) moves in the CCW direction and all other robots on the CCW chain move one hop in the CCW direction. So in configuration $C(t)$, $d_{CCW}(v_0,v_1)=1$, which is less than $k$ ( as the configuration is a non-chain configuration which is not target configuration too), where $v_1$ is the next node of $v_0$ in the counterclockwise direction. Again, for all other consecutive nodes $v_x$, $v_y$ in the chain, $d_{CCW}(v_x,v_y)=k$.

Let $r_p$ be the last  robot on counterclockwise chain in $C(t-1)$ located at $v_{p-1}$ and $r_q$ be the robot on adjacent occupied node of $v_{p-1}$ in $C(t-1)$. Let the position of $r_p$ in $C(t)$ be $v_p$. Then $d_{CCW}(v_{p-1}, v_p)=1$. Let $v_q$ be the adjacent occupied node of $r_p$ in the counter-clockwise direction in both $C(t-1)$ and $C(t)$ (as the robot $r_q$ on the node $v_q$ belongs to the clockwise chain and no robots on clockwise chain moves in $C(t-1)$). To show that $C(t)$ contains a counter-clockwise block, we have to show that, $d_{CCW}(v_p,v_q)>k$ in $C(t)$. Since $C(t-1)$ is a chain configuration $d_{CCW}(v_{p-1},v_q) >k \implies d_{CCW}(v_p,v_q)\ge k$. We show that if $d_{CCW}(v_p,v_q)= k$ then we arrive at a contradiction. If  $d_{CCW}(v_p,v_q)= k$ then, in $C(t)$, except robots at $v_0$ and $v_1$, all other pairs of consecutive occupied nodes will have distance $k$ between them. Let $l'$ be the number of robot positions in $C(t)$ then $n=1+(l'-1)k$. In target configuration distance between any two adjacent occupied nodes must have to be at least $k$. So $n\ge lk$ where $l$ is the total number of robots deployed on the ring. So, $1+(l'-1)k \ge lk \ge l'k \implies k \le 1$. This contradicts the fact that $k>1$. So, $d_{CCW}(v_p,v_q)> k$. 

Then in configuration $C(t)$, a block occurs in the CCW direction which is a non-chain block. Also since the distance between all adjacent occupied nodes in the arc $(v_0,v_q)_{CW}$ is exactly $k$ and $d_{CW}(v_q,v_p)>k$ in $C(t)$ (as argued above). So, a block occurs in the clockwise direction too which is a chain block by the definition~\ref{defn:block}. 

\textit{Case I(b):} One chain is good and another is bad in the chain configuration $C(t-1)$. Let $\mathcal{D}$ be the direction (either clockwise or counter-clockwise) in which the good chain occurs. Then by Algorithm~\ref{Spread}, the least ID robot on the multiplicity node moves once in direction $\mathcal{D}$ and all other robots on the chain in direction $\mathcal{D}$ move once in the direction $\mathcal{D}$. Thus by similar argument, in configuration $C(t)$, a non-chain block occurs in direction $\mathcal{D}$ and a chain block occurs in the direction $\mathcal{D}'$, which is the opposite direction of $\mathcal{D}$.

\textit{Case I(c):} In the chain configuration $C(t-1)$, both chains are bad chain. In this case, the missing edge must be in between the multiplicity and its immediate next node in some direction, say $\mathcal{D}$ (either clockwise or counter-clockwise). Then executing Algorithm~\ref{Spread}, the least ID robot on the multiplicity node moves once in direction $\mathcal{D}$ and all other robots on the chain in direction $\mathcal{D}$ move once in the direction $\mathcal{D}$. Thus in this case also, the configuration $C(t)$ must have two blocks: a non-chain block in direction $\mathcal{D}$ and a chain block in the opposite direction $\mathcal{D}'$.

 Note that from the above argument, the first non-chain configuration will have two blocks. We now prove that for any non-chain configuration having two blocks if the next configuration is again a non-chain configuration it will have two blocks if the next configuration is not the target configuration. This will be enough to prove that, while executing Algorithm~\ref{DkD_DYNAMICRING}, any non-chain configuration will have two blocks. 
 We discuss this in the next case.
 
\textit{Case II:} Let $C(t-1)$ be a non-chain configuration which has two blocks. We have to show that the configuration $C(t)$ also has two blocks if $C(t)$ is not a chain configuration. As $C(t-1)$ is a non-chain configuration with two blocks then two cases can arise :

\textit{Case II(a):} The missing edge in non-chain block or no missing edge. Let the chain block be in the direction $\mathcal{D}$. According to the algorithm~\ref{ConstructChain} the robots on multiplicity move in the direction $\mathcal{D}$ and the singleton robots on chain block move away from the block head in direction $\mathcal{D}$. Thus in $C(t)$ distance between any two adjacent occupied nodes remains the same if those nodes were in chain block in $C(t-1)$. Let $v_0$ be the block head in $C(t-1)$ and $v_1$ be the adjacent occupied node of $v_0$ in direction $\mathcal{D}'$ in $C(t-1)$, where $\mathcal{D}'$ is the opposite direction of $\mathcal{D}$. Let $v_{i+1}$ be the adjacent occupied node of $v_{i}$ in direction $\mathcal{D}'$ in $C(t-1)$ where $i\in \{1,2, \cdots, p-1\}$ for some $p\in \mathbb{N} $. Note that in $C(t)$ the robots which were on $v_0$ in $C(t-1)$ move one hop in direction $\mathcal{D}$. So in $C(t)$, the sequence of distances between adjacent occupied nodes in direction $\mathcal{D}'$ starting from the robot position containing the robots which were forming the block head in $C(t-1)$ up to the robot position at $v_{p}$ is $k_0, k, k, \dots k($ total $ p $ terms), where $k_0  \le k.$ Also let $r_p$ and $r_q$ be the last robots on the non-chain and chain blocks in $C(t-1)$ located at vertices, say $v_p$ and $v_q$. Then in $C(t-1)$, distance between $v_p$ and $v_q$ is strictly greater than $k$. In $C(t)$ the position of $r_p$ and $r_q$ is either $v_{p+1}$ and $v_q$ or $v_p$ and $v_{q+1}$ where $v_{p+1}$ and $v_{q+1}$ are the adjacent nodes of $v_p$ and $v_q$ in the direction of shortest distance between $r_p$ and $r_q$ in $C(t)$. Now for both cases, we can argue that the shortest distance between $r_p$ and $r_q$ is strictly greater than $k$ in $C(t)$ (argument similar as in \textit{Case I(a)}). This implies that $C(t)$ has two blocks again. One of them is a chain block and the other one is a non-chain block. 

\textit{Case II(b):} There is missing edge in chain block. In this case, executing the algorithm~\ref{ConstructChain}, all singleton robots in non-chain block move in some direction $\mathcal{D}$ away from the block head. Let the sequence of adjacent occupied nodes in direction $\mathcal{D}$ that is in the non-chain block in $C(t-1)$ is $v_0,v_1,\dots, v_p$. This way the sequence of adjacent occupied nodes in direction $\mathcal{D}$ in $C(t)$ for all robots that were on the non-chain block of $C(t-1)$ is $v_0, v_1',v_2'\dots,v_p'$ where $v_i'$ is the adjacent node of $v_i$ in direction $\mathcal{D}$ for $i \in \{1,2,\dots, p\}$. Then The sequence of distances of adjacent occupied nodes in direction $\mathcal{D}$ starting from $v_0$ up to $v_p'$ in $C(t)$ is $k_0,k,k,\dots, k (p+1$ terms$)$ where $k_0 \le k$. Now similar to the argument for the above case, it can be shown that $C(t)$ has two blocks one is a chain block and the other is a non-chain block. 
\end{proof}

\begin{lemma}\label{lemma: non chain configuration 1st time with d(H,H+1)=1.}

Let $C(t)$ be a chain configuration while $C(t+1)$ be a non-chain configuration during the execution of Algorithm~\ref{DkD_DYNAMICRING}. Then the following statements are true.
\begin{enumerate}
    \item In $C(t+1)$, there exists a direction $\mathcal{D}$ such that $d_{\mathcal{D}}(H, H+1)=1$, where $H$ and $H+1$ are the block head and the adjacent occupied node of block head in direction $\mathcal{D}$
    \item $d_{\mathcal{D}}(v_x,v_y)=k$, for all consecutive occupied nodes on the block in direction $\mathcal{D}$ where $v_x \ne H$.
    \item  $d_{\mathcal{D'}}(v_x,v_y) =k$ for all consecutive occupied nodes $v_x$ and $v_y$ on the block in the direction $\mathcal{D}'$, where $\mathcal{D}'$ be the opposite direction of $\mathcal{D}$.
\end{enumerate}
\end{lemma}
\begin{proof}
Let for some round $t \ge 0$ the configuration $C(t)$ be a chain configuration and during the execution of Algorithm~\ref{DkD_DYNAMICRING} at round $t$, the configuration $C(t+1)$ becomes non-chain configuration. 
In $C(t)$, the distance between any two consecutive occupied nodes in the same chain is equal to $k$ (By Definition~\ref{defn:chain config}). Now by Algorithm~\ref{DkD_DYNAMICRING}, all agents in $C(t)$ executes Algorithm~\ref{Spread}. By Lemma~\ref{lemma: non chain config contains two blocks: chain block and non-chain block}, $C(t+1)$ must have two blocks one chain block and one non-chain block. Note that in $C(t)$ there can be three cases, 

\textit{Case I:} Both chains are good chain in $C(t)$. In this case, the least ID robot, say $r$ on multiplicity, say $v$, moves counter-clockwise to the adjacent node, say $v'$ and all other singleton nodes on the counter-clockwise good chain in $C(t)$ moves counter-clockwise once. So, now in $C(t+1)$, $v$ and $v'$ both have robots. Also, let in $C(t)$, $r_i$, and $r_j$ be two robots on the counter-clockwise chain such that they are located on $v_i$ and $v_j$, where $v_i$ and $v_j$ are two consecutive occupied singleton nodes. Then in $C(t+1)$, $r_i$ and $r_j$ are located on $v_i'$ and $v_j'$ respectively, where $v_i'$ and $v_j'$ are the counterclockwise adjacent nodes of $v_i$ and $v_j$ respectively. Now, note that in $C(t+1)$, $d_{CCW}(r_i,r_j)= d_{CCW}(v_i',v_j') = d_{CCW}(v_i,v_j)= k$. Thus we prove the first and second statements for this case.  
Now note that in the clockwise chain of $C(t)$ all the robot positions stay the same. So third statement follows directly. 

\textit{Case II:} In $C(t)$ one chain is good and one is bad. Let $\mathcal{D}$ be the direction in which the good chain occurs i.e. the missing edge does not happen to be in the chain in direction $\mathcal{D}$. Then in this case, The least ID robot, say $r$, on the multiplicity, say $v$, moves one hop in direction $\mathcal{D}$ and all the robots on singleton nodes of the good chain in $C(t)$ moves one hop in direction $\mathcal{D}$. Also no robot on the bad chain in $C(t)$ moves (except the least ID robot that moves in direction $\mathcal{D}$). So by a similar argument for the case above, we can conclude all three statements.

\textit{Case III:} In $C(t)$ both chains are bad chain. Then the missing edge is incident to the multiplicity in some direction, say $\mathcal{D}$. Then in the direction $\mathcal{D}'$, the opposite direction of $\mathcal{D}$, multiplicity has no missing edge. Then in this case the least ID robot, say $r$, on multiplicity moves once in direction $\mathcal{D}'$ and all the robots on singleton nodes of the chain in direction $\mathcal{D}'$ move once in direction $\mathcal{D}'$. Then by the similar argument as above, all three statements can be concluded.
\end{proof}

\begin{lemma}\label{lemma: non chain config with d(H,H+1)=d, next round distance increases 1.}

Let for some $t>0$, $C(t)$ be a non-chain configuration that is not the target configuration. Let $\mathcal{D}$ be the direction of the non-chain block in $C(t)$. Let $d_{\mathcal{D}}(H,H+1)=d <k$ where $H$ is block head and $H+1$ is the  adjacent occupied node of $H$ in direction $\mathcal{D}$. Then $d_{\mathcal{D}}(H',H'+1)=d+1$ in $C(t+1)$. Here $H'$ is either block head or multiplicity in $C(t+1)$ and $H'+1$ is the adjacent occupied node of $H'$ in direction $\mathcal{D}$ in $C(t+1)$.  
\end{lemma}
\begin{proof}
Let $C(t)$ be a non-chain configuration that is not the target configuration. then by Lemma~\ref{lemma: non chain config contains two blocks: chain block and non-chain block}, $C(t)$ must have two blocks, one chain block, and one non-chain block. Let $H$ be the block head in $C(t)$. Note that the distance from $H$ to its adjacent occupied node in the direction of the chain block say $\mathcal{D}_1$, must be $k$. Let $\mathcal{D}$ be the direction of the non-chain block in $C(t)$. So if $\mathcal{D}$ be the direction for which $d_{\mathcal{D}}
(H, H+1)=d<k$  in $C(t)$, where $H+1$ is the adjacent occupied node of $H$, then $\mathcal{D}$ must be same as $\mathcal{D}_1'$. Now as $C(t)$ is a non-chain configuration then three cases can occur.

\textit{Case I:} Let the missing edge be on the chain block. In this case, according to the algorithm~\ref{ConstructChain}, all singleton robots on the non-chain block move away from the block head $H$ in the direction of $\mathcal{D}$ i.e. the direction of the non-chain block. So the singleton robot on $H+1$, the adjacent occupied node of the block head $H$, moves one hop distance in direction $\mathcal{D}$ in the configuration $C(t)$. If in configuration $C(t+1)$, $H'$ and $H'+1$ be the new positions of the block head and adjacent occupied node of block head in direction $\mathcal{D}$, then position $H$ in $C(t)$ is same as the position $H'$ in $C(t+1)$ and the position $H'+1$ in $C(t+1)$ is one hop ahead from the position $H+1$ in the configuration $C(t)$. Then in configuration $C(t+1)$, the distance between $H'$ and $H'+1$ will be $d+1$ in the direction $\mathcal{D}$.

\textit{Case II:} Let the missing edge be on the non-chain block. Then according to the algorithm~\ref{ConstructChain}, the singleton robots on the chain block and the robots on the block head, positioned at $H$, move once in the direction of the chain block $\mathcal{D}_1$ in the configuration $C(t)$. As the robots on the non-chain block does not move in $C(t)$, then in $C(t)$ the position $H+1$ (the position of the adjacent occupied node of block head in direction $\mathcal{D}$ in $C(t)$) is same as the position $H'+1$ (the position of the adjacent occupied node of block head in direction $\mathcal{D}$ in $C(t+1)$) in the configuration $C(t+1)$. But in configuration $C(t+1)$, the new position of the block head $H'$ is one hop away from the position of the block head $H$ in $C(t)$ in the direction $\mathcal{D}_1$ (same as $\mathcal{D}'$). So in the configuration $C(t+1)$, the distance between $H'$ and $H'+1$ in direction $\mathcal{D}$ will be increased by $d+1$.

In this case, all singleton robots on the chain block move away from $H$ and the robots on $H$ move towards the chain block. So the distance between $H$ and $H+1$ becomes $d+1$ in the configuration $C(t+1)$. 

\textit{Case III:} Let there be no missing edge in the non-chain configuration. The argument for this case is similar to the argument for \textit{Case I}. 
\end{proof}

\begin{lemma}\label{lemma: non-chain config to chain-config occurs in at most k-1 round.}

Let $C(t)$ be a non-chain configuration for some $t>0$, which is not the target configuration. Let $C(t')$ be the first chain configuration (if exists) or the target configuration after $C(t)$, where $t'>t$. Then $t'-t \le k-1$.
\end{lemma}
\begin{proof}
Let the configuration $C(t)$ be a non-chain configuration at round $t>0$, which is not the target configuration. Then $C(t)$ has two blocks: one chain block and another non-chain block by lemma\ref{lemma: non chain config contains two blocks: chain block and non-chain block}. Let $\mathcal{D}$ be the direction of the non-chain block. Then $d_{\mathcal{D}}(H,H+1)=d<k$, where $H $ is block head in $C(t)$ and $H+1$ is the adjacent occupied node in direction $\mathcal{D}$ in $C(t)$. Let $k-d=i>0$. $i$ can be at most $k-1$. If $C(t+1)$ is the target configuration or a chain configuration then we are done. Otherwise, the distance between the position of the new block head, say $H'$, and the adjacent occupied node of $H'$ in the direction of the non-chain block, say $H'+1$, becomes $d+1$ in $C(t+1)$ (by Lemma~\ref{lemma: non chain config with d(H,H+1)=d, next round distance increases 1.}). This way until the configuration becomes a chain configuration or the target configuration, the distance between the block head and the adjacent occupied node in the direction of the non-chain block increases in each round. thus note that in $C(t+i)$ the distance between each pair of adjacent occupied nodes becomes at least $k$. Thus $C(t+i)$ is either the target configuration (if $C(t+i)$ has no multiplicity)or, it becomes a chain configuration. So let us take $t'=t+i$ thus, $t'-t=i \le k-1$. This concludes the claim.

\end{proof}

\begin{theorem}\label{th:f}
The algorithm \textsc{D-$k$-D DynamicRing} will terminate in $O(n)$ rounds.  
\end{theorem}
\begin{proof}Let the initial configuration $C(0)$ has $l$ many robots  on the multiplicity node. Also in $C(0)$, there are no other robot positions except the multiplicity node. So, $C(0)$ is a chain configuration by Definition~\ref{defn:chain config}. Now according to Algorithm~\ref{DkD_DYNAMICRING}, the robots at round 0 execute the subroutine \textsc{Spread}. So $C(1)$ is either a chain configuration or a non-chain configuration having two blocks, one is a chain block and the other is a non-chain block (Lemma~\ref{lemma: non chain config contains two blocks: chain block and non-chain block}). Also by Lemma~\ref{lemma:singleton node increases in spread}, the singleton node increases from 0 to 1 in $C(1)$. Thus, $C(1)$ has two robot positions one is a multiplicity with $l-1$ robots and the other one is a singleton node. If $C(1)$ is not a chain configuration, then by Lemma~\ref{lemma: non-chain config to chain-config occurs in at most k-1 round.}, within at most $k-1$ rounds the configuration becomes a chain configuration having two robot positions. This is because, during \textsc{ReconstructChain} subroutine either all robots on multiplicity move together in the same direction or they do not move at all. Also, no singleton robots collide to decrease robot positions as all singleton robots on exactly one block (either clockwise or, in a counter-clockwise direction) move in the same direction and the last robots on two blocks do not collide as the distance between them is strictly greater than $k \ge 1$. this way it can be concluded that if for some $t\ge 0$, $C(t)$ is a chain configuration where the multiplicity node has $2 \le l'\le l$ robots and $l-l'$ singleton nodes, then there exists a $t'>t$ such that, $C(t')$ is a chain configuration where the multiplicity node contains $l'-1$ robots and $l-l'+1$ singleton nodes (Lemma~\ref{lemma:singleton node increases in spread} and Lemma~\ref{lemma: non-chain config to chain-config occurs in at most k-1 round.}) where $t' \le t+k-1$. This way there exists a round $t_0\le (l-2)k$ such that, the configuration $C(t_0)$ becomes a chain configuration where the multiplicity node has two robots and $l-2$ singleton nodes. From this configuration, in 1 round the configuration $C(t_0+1)$ becomes a dispersed configuration (due to execution of \textsc{Spread} subroutine). Now in the worst case, $C(t_0+1)$ is not the target configuration. Then again by Lemma~\ref{lemma: non chain config contains two blocks: chain block and non-chain block} must have two blocks. Also note that By Lemma~\ref{lemma: non-chain config to chain-config occurs in at most k-1 round.}, there exists $t_1 \le t_0+k$ such that $C(t_1)$ is the target configuration (As $C(t_0')$ can not be a chain configuration anymore where $t_0' > t_0$ as no collision occurs to create further multiplicity). Now since $t_1 \le  t_0+k\le (l-2)k+k=(l-1)k$, The algorithm terminates within $(l-1)k$ rounds. Now $l$ can be at most $\lfloor \frac{n}{k}\rfloor$. So the Algorithm~\ref{DkD_DYNAMICRING} terminates within $(l-1)k < lk = \lfloor \frac{n}{k}\rfloor k \le \frac{n}{k}k=n$. Thus, Algorithm \textsc{D-$k$-D DynamicRing} terminates in $O(n)$ rounds.
\end{proof}

From Theorem~\ref{th:f} and Corollary~\ref{th:f'} we have the following result:

\begin{theorem}
    Algorithm \textsc{D-$k$-D DynamicRing} terminates in $\Theta(n)$ rounds for any 1-interval connected ring with $n$ nodes.
\end{theorem}

\subsection{Algorithm without chirality}
If robots have no chirality agreement, then by algorithm \emph{No-Chir-Preprocess} proposed in \cite{AAMKS2018ICDCN} by Agarwalla et al. can be used by $l$ co-located robots to agree on a particular direction before starting the execution of \textsc{D-$k$-D DynamicRing}. This way we can solve the D-$k$-D on a dynamic ring even when the robots do not agree on chirality.

\section{Conclusion}\label{Section:7}
The primary aim of this paper is to introduce the D-$k$-D problem, a generalization of all previously studied variants of the dispersion problem for mobile robots. This problem is examined within the context of a 1-interval connected ring under a fully synchronous scheduler. The paper begins by discussing the necessity of a fully synchronous scheduler to solve this problem, providing a lower bound on the time required, which is $\Omega(n)$ for a dynamic ring with $n$ nodes. Following this, the paper presents an algorithm, \textsc{D-$k$-D DynamicRing}, which solves the problem in optimal time for a rooted initial configuration, assuming chirality. Additionally, it is noted that for robots without chirality agreement, the problem can still be solved using a known technique (\cite{AAMKS2018ICDCN}).

Looking ahead, several open problems related to this variant of dispersion present opportunities for further research, especially with the goal of solving this problem in arbitrary networks. Interestingly, when extending the technique used to solve D-2-D in \cite{KM2023NETYS} by Kaur et al., the time and memory complexity becomes exponential even when $k=3$, making it an intriguing subject for study in arbitrary networks. Another potential research direction involves examining this variant under a limited visibility model, where a robot can only see up to a distance $\phi < \frac{n}{2}$ from its position on a dynamic ring.

\bibliographystyle{abbrv}
\bibliography{samplepaper}
\end{document}